\documentclass[a4paper,11pt]{article}

\usepackage{amsmath}
\usepackage{amsfonts}
\usepackage{amssymb}
\usepackage{caption}
\usepackage{graphicx}
\usepackage{graphics}
\usepackage{color}
\usepackage[numbers]{natbib}
\usepackage{ushort}
\usepackage{makeidx}
\usepackage{float}
\usepackage{appendix}
\usepackage{framed}
\usepackage{versions}
\usepackage{xcolor}
\usepackage{mathbbol}
\usepackage{blindtext}
\usepackage{emptypage}
\usepackage{verbatim}
\usepackage{enumitem}
\usepackage{array}
\usepackage{multirow}

\def\hmath$#1${\texorpdfstring{{\rmfamily\textit{#1}}}{#1}}

\usepackage[bookmarks]{hyperref}
\hypersetup{
    colorlinks=true,   	
    linkcolor=red,      
    citecolor = [rgb]{0 0.7 0},   	
    filecolor=magenta, 	
    urlcolor=blue
}

\newcommand{\Xp}{\mbox{\boldmath $X$}}

\newcommand{\Xps}{\mbox{\scriptsize\boldmath $X$}}
\newcommand{\xp}{\mbox{\boldmath $x$}}
\newcommand{\xps}{\mbox{\scriptsize\boldmath $x$}}
\newcommand{\yps}{\mbox{\scriptsize\boldmath $y$}}

\newcommand{\Yp}{\mbox{\boldmath $Y$}}
\newcommand{\Yps}{\mbox{\scriptsize\boldmath $Y$}}
\newcommand{\yp}{\mbox{\boldmath $y$}}
\newcommand{\Zp}{\mbox{\boldmath $Z$}}
\newcommand{\Zps}{\mbox{\scriptsize\boldmath $Z$}}
\newcommand{\Zpss}{\mbox{\tiny\boldmath $Z$}}

\newcommand{\leqa}{\mbox{$ \;\stackrel{(a)}{\leq}\; $}}
\newcommand{\leqb}{\mbox{$ \;\stackrel{(b)}{\leq}\; $}}

\newcommand{\IN}{{\mathbb Z}}

\newcommand{\BBP}{{\mathbb P}}

\newcommand{\BBE}{{\mathbb E}}

\newcommand{\VAR}{\mbox{\rm Var}}

\def\ba{\begin{align}}
\def\ea{\end{align}}
\def\ban{\begin{align*}}
\def\ean{\end{align*}}

\def\be{\begin{eqnarray}}
\def\ee{\end{eqnarray}}
\def\ben{\begin{eqnarray*}}
\def\een{\end{eqnarray*}}

\def\bqq{\begin{equation}}
\def\eqq{\end{equation}}
\def\bqqn{\begin{equation*}}
\def\eqqn{\end{equation*}}

\def\elabel#1{\label{e:#1}}

\def\sq{$\Box$}

\def\qed{\ifmmode\sq\else{\unskip\nobreak\hfil
\penalty50\hskip1em\null\nobreak\hfil\sq
\parfillskip=0pt\finalhyphendemerits=0\endgraf}\fi\par\medbreak}

\newsavebox{\junk}
\savebox{\junk}[1.6mm]{\hbox{$|\!|\!|$}}

\def\limsup{\mathop{\rm lim\ sup}}
\def\liminf{\mathop{\rm lim\ inf}}

\def\til={{\widetilde =}}

\def\clR{{\cal R}}

\def\clX{{\cal X}}
\def\clY{{\cal Y}}
\def\clZ{{\cal Z}}

 \def\eq#1/{(\ref{#1})}

\newtheorem{theorem}{Theorem}[section]
\newtheorem{corollary}[theorem]{Corollary}

\newtheorem{lemma}[theorem]{Lemma}
\newtheorem{definition}[theorem]{Definition}

\def\eq#1/{(\ref{e:#1})}

\newcommand{\beqn}[1]{\notes{#1}%
\begin{eqnarray} \elabel{#1}}

\newcommand{\eeqn}{\end{eqnarray} } 

\newcommand{\beq}[1]{\notes{#1}%
\begin{equation}\elabel{#1}}

\newcommand{\eeq}{\end{equation}} 

\def\bdes{\begin{description}}
\def\edes{\end{description}}

\DeclareMathOperator*{\Bcup}{\text{\raisebox{0.25ex}%
	{\scalebox{0.8}{$\bigcup$}}}}

\def\notes#1{}

\definecolor{mag}{rgb}{0.7,0,0.3}
\definecolor{dgreen}{rgb}{0.1,0.5,0.1}
\definecolor{dred}{rgb}{.8,0,0}
\definecolor{gray}{rgb}{.8,.8,.8}
\definecolor{brown}{rgb}{0.6451,0.3706,0.1745}

\newcommand{\E}{\mathbb{E}}
\newcommand{\PP}{\mathbb{P}}
\newcommand{\Pbig}[1]{%
    \PP\bigl[ #1\bigr]%
}
\newcommand{\seq}[1]{%
    \{ #1\}%
}

\newcommand{\Ebig}[1]{%
    \E\bigl( #1\bigr)%
}
\newcommand{\EBig}[1]{%
    \E\Bigl( #1\Bigr)%
}

\newcommand{\PBig}[1]{%
    \PP\Bigl[ #1\Bigr]%
}

\newenvironment{proof}{\paragraph{Proof. }}{\hfill$\square$}

\begin{document}

\title{\vspace{-1.5cm}%
Sharp Second-Order Pointwise Asymptotics\\
for Lossless Compression with Side Information}

\author
{
	Lampros Gavalakis
    \thanks{Department of Engineering,
	University of Cambridge,
        Trumpington Street,
	Cambridge CB2 1PZ, U.K.
                Email: \texttt{\href{mailto:lg560@cam.ac.uk}%
			{lg560@cam.ac.uk}}.
	L.G.\ was supported in part by EPSRC grant number RG94782.
        }
\and
        Ioannis Kontoyiannis 
    \thanks{Department of Engineering,
	University of Cambridge,
        Trumpington Street,
	Cambridge CB2 1PZ, U.K.
                Email: \texttt{\href{mailto:i.kontoyiannis@eng.cam.ac.uk}%
			{i.kontoyiannis@eng.cam.ac.uk}}.
		Web: \texttt{\url{http://www.eng.cam.ac.uk/profiles/ik355}}.
	I.K.\ was supported in part by a grant from the Hellenic Foundation
	for Research and Innovation.
        }
}

\date{\today}

\maketitle

\setcounter{page}{0}
\thispagestyle{empty}

\begin{abstract}
The problem of determining the best achievable
performance of arbitrary lossless compression
algorithms is examined, when 
correlated side information is available at 
both the encoder and decoder.
For arbitrary source-side information pairs,
the conditional information density
is shown to provide a sharp asymptotic
lower bound for the description lengths 
achieved by an arbitrary sequence of compressors.
This implies that,
for ergodic source-side information pairs,
the conditional entropy rate is 
the best achievable asymptotic lower bound
to the rate,
not just in expectation but with probability
one. Under appropriate mixing conditions,
a central limit theorem and a
law of the iterated logarithm are proved,
describing the inevitable fluctuations 
of the second-order asymptotically best 
possible rate.
An idealised version of Lempel-Ziv coding with 
side information is shown to be
universally first- and second-order asymptotically 
optimal, under the same conditions.
These results are in part based on a new
almost-sure invariance principle for the
conditional information density, which may
be of independent interest.
\end{abstract}



\newpage

\setcounter{page}{1}

\section{Introduction}

It is well-known that the presence of correlated side 
information can potentially offer 
dramatic benefits for data compression
\cite{slepianwolf:73,cover:book2}.
Important applications where such side information 
is naturally present include
the compression of genomic data \cite{yang:01,
fritz:11},
file and software management \cite{rsync,suel:02},
and image and video compression
\cite{pradhan:01,aaron:02}.

In practice, the most common approach
to the design of effective 
compression methods with side
information is based on 
generalisations of the Lempel-Ziv
family of algorithms
\cite{subrahmanya:95,uyematsu:03,tock:05,jacob:08,jain-bansal}.
A different approach based on 
grammar-based codes was developed in~\cite{stites:00},
turbo codes were applied in~\cite{aaron:02b},
and a generalised version of context-tree weighting
was used in~\cite{C-K-Verdu:06}.

In this work we examine the theoretical fundamental 
limits of the best possible performance that can be 
achieved in such problems.
Let $(\Xp,\Yp)=\{(X_n,Y_n)\;;\;n\geq 1\}$ be a source-side
information pair; $\Xp$ is the source to be compressed,
and $\Yp$ is the associated side information process
which is assumed to be available both to the encoder
and the decoder.
Under appropriate conditions, the best {\em average}
rate that can be achieved asymptotically \cite{cover:book2},
is the {\em conditional entropy rate},
$$H(\Xp|\Yp)=\lim_{n\to\infty}\frac{1}{n}H(X_1^n|Y_1^n),\
\qquad\mbox{bits/symbol},$$ 
where 
$X_1^n=(X_1,X_2,\ldots,X_n)$,
$Y_1^n=(Y_1,Y_2,\ldots,Y_n)$,
and $H(X_1^n|Y_1^n)$ denotes the conditional entropy
of $X_1^n$ given $Y_1^n$;
precise definitions will be given in 
Section~\ref{pointwiseasymptotics}.

Our main goal is to derive sharp
asymptotic expressions for the optimum compression 
rate (with side information available
to both the encoder and decoder), 
not only in expectation but with probability~1.
In addition to 
the best first-order performance,
we also determine 
the best rate at which this performance
can be achieved, as a function of the length
of the data being compressed. Furthermore, we consider
an idealised version of a Lempel-Ziv 
compression algorithm,
and we show that it can achieve 
asymptotically optimal first- and second-order
performance, {\em universally} over a broad
class of stationary and ergodic source-side 
information pairs $(\Xp,\Yp)$.

Specifically, we establish the following.
In Section~\ref{descriptionofanoptimal} we
describe the theoretically optimal 
one-to-one compressor $f_n^*(X_1^n|Y_1^n)$, 
for arbitrary source-side information pairs
$(\Xp,\Yp)$.
In Section~\ref{s:stronga} we 
prove our first result, stating that 
the description lengths $\ell(f_n^*(X_1^n|Y_1^n))$
can be well-approximated, with probability
one, by the {\em conditional information
density}, $-\log P(X_1^n|Y_1^n)$.
Theorem~\ref{t:kieffer} 
states that,
for any jointly stationary and ergodic
source-side information pair $(\Xp,\Yp)$,
the best asymptotically achievable compression 
rate is $H(\Xp|\Yp)$ bits/symbol,
with probability~1.
This generalises Kieffer's corresponding result~\cite{kieffer:91} 
to the case of compression with side information.

Further, in Section~\ref{s:finer}
we show that
there is a sequence of random variables $\{Z_n\}$ such that
the description lengths $\ell(f_n(X_1^n|Y_1^n))$
of {\em any} sequence of compressors $\{f_n\}$ satisfy a
``one-sided'' central limit theorem (CLT): 
Eventually, with probability~1,
\be
\ell(f_n(X_1^n|Y_1^n))\geq nH(\Xp|\Yp)+\sqrt{n} Z_n
+o(\sqrt{n}),\qquad\mbox{bits},
\label{eq:ICLT}
\ee
where the $Z_n$ converge to a $N(0,\sigma^2(\Xp|\Yp))$
distribution, and the term $o(\sqrt{n})$ is 
negligible compared to $\sqrt{n}$. The lower bound~(\ref{eq:ICLT})
is established in Theorem~\ref{t:cltc} where it is also
shown that it is asymptotically achievable.
This means that the rate obtained by any sequence
of compressors has inevitable $O(\sqrt{n})$ fluctuations
around the conditional entropy rate,
and that
the size of these fluctuations is quantified
by the {\em conditional varentropy rate},
$$
\sigma^2(\Xp|\Yp)=\lim_{n\to\infty}\frac{1}{n}
\VAR\big(-\log P(X_1^n|Y_1^n)\big).$$
This generalises the {\em minimal coding variance}
of \cite{kontoyiannis-97}.
The bound~(\ref{eq:ICLT}) holds
for a broad class of source-side information pairs,
including all Markov chains with positive
transition probabilities. 
Under the same conditions, a corresponding 
``one-sided'' law of the iterated logarithm (LIL)
is established in Theorem~\ref{t:lilc}, 
which gives a precise description of the
inevitable almost-sure fluctuations 
above $H(\Xp|\Yp)$, for 
any sequence of compressors.

The proofs of all the results in
Sections~\ref{s:first} and~\ref{s:finer}
are based, in part, on analogous
asymptotics for
the conditional information density,
$-\log P(X_1^n|Y_1^n)$. These 
are established in Section~\ref{s:density},
where we state and prove
a corresponding CLT and an LIL for $-\log P(X_1^n|Y_1^n)$.
These results, in 
turn, follow from the almost sure
invariance principle for $-\log P(X_1^n|Y_1^n)$,
proved in Appendix~\ref{philippstoutproof}.
Theorem~\ref{philippstout},
which is of independent interest,
generalises the invariance
principle established for the (unconditional)
information density $-\log P(X_1^n)$ by 
Philipp and Stout~\cite{philipp-stout:book}.

Section~\ref{lempelzivsection} is devoted to 
{\em universal} compression.
We consider a simple, idealised version of
Lempel-Ziv coding with side information. As in the
case of Lempel-Ziv compression without side information
\cite{willems:1,wyner-ziv:1}, 
the performance of this scheme is 
determined by the asymptotics of a family
of {\em conditional recurrence times}, $\clR_n=\clR_n(\Xp|\Yp)$.
Under appropriate, general conditions on the
source-side information pair $(\Xp,\Yp)$,
in Theorem~\ref{strongapproxtheorem} we show
that the ideal description lengths,
$\log \clR_n$, can be well-approximated
by the conditional information
density $-\log P(X_1^n|Y_1^n)$.
Combining this with our earlier results
on the conditional information density,
in Corollary~\ref{strongapproxcorollary}
and 
Theorem~\ref{recurrenceinvariance}
we show that the compression rate
of this scheme converges to $H(\Xp|\Yp)$,
with probability~1,
and that it is universally
second-order optimal.
The results of this section generalise
the corresponding asymptotics without side
information established in
\cite{ornstein-weiss:2} and \cite{kontoyiannis-jtp}.

The proofs of the more technical results needed in
Sections~\ref{pointwiseasymptotics} 
and~\ref{lempelzivsection} are given in the appendix.

\section{Pointwise Asymptotics}
\label{pointwiseasymptotics}

In this section we derive general, 
fine asymptotic bounds for the description lengths
of arbitrary compressors with side 
information, as well as corresponding
achievability results.

\subsection{Preliminaries}
\label{descriptionofanoptimal}

Let $\Xp=\{X_n\;;\;n\geq 1\}$ be an arbitrary source
to be compressed,
and $\Yp=\{Y_n\;;\;n\geq 1\}$ be an associated 
side information process. We let 
$\clX,\clY$, denote their finite alphabets,
respectively, and we refer to the joint process
$(\Xp,\Yp)=\{(X_n,Y_n)\;;\;n\geq 1\}$ as a
{\em source-side information pair}.

Let $x_1^n=(x_1,x_2,\ldots,x_n)$ be a source string,
and let $y_1^n=(y_1,y_2,\ldots,y_n)$ an associated 
side information string which is available to both the encoder
and decoder.
A {\em fixed-to-variable one-to-one compressor 
with side information}, of blocklength $n$,
is a collection of functions $f_n$, where each
$f_n(x_1^n|y_1^n)$ takes a value in the set of all finite-length
binary strings,
$$\{0,1\}^*=\Bcup_{k=0}^\infty\{0,1\}^k=\{\emptyset,0,1,00,01,000,\ldots\},$$
with the convention that $\{0,1\}^0=\{\emptyset\}$ consists 
of just the empty string
$\emptyset$ of length zero. 
For each $y_1^n\in\clY^n$, we assume that
$f_n(\cdot|y_1^n)$ is a one-to-one function 
from $\clX^n$ to $\{0,1\}^*$,
so that the compressed binary string $f_n(x_1^n|y_1^n)$ 
is uniquely decodable. 

The main figure of merit in lossless compression
is of course the description length,
$$\ell(f_n(x_1^n|y_1^n))
=\mbox{length of}\;f_n(x_1^n|y_1^n),\qquad\mbox{bits},$$
where, throughout, $\ell(s)$ denotes the length, in bits, of 
a binary string $s$.
It is easy to see that, under quite general criteria,
the optimal compressor $f_n^*$ is easy
to describe;
see \cite{gavalakis-arxiv:19} for an extensive 
discussion.
For $1\leq i\leq j\leq\infty$, we use the shorthand 
notation $z_i^j$ for the string $(z_i,z_{i+1},\ldots,z_j)$,
and similarly $Z_i^j$ for the corresponding collection 
of random variables $Z_i^j=(Z_i,Z_{i+1},\ldots,Z_j)$.

\begin{definition}
[The optimal compressor \boldmath{$f_n^*$}]
For each side information string $y_1^n$,
$f_n^*(\cdot|y_1^n)$ is the optimal compressor for 
the distribution $\BBP(X_1^n=\cdot|Y_1^n=y_1^n)$, namely, 
the compressor that orders the strings $x_1^n$
in order of decreasing probability $\BBP(X_1^n=x_1^n|Y_1^n=y_1^n)$, 
and assigns them codewords from $\{0,1\}^*$ in lexicographic order.
\end{definition}

\subsection{The conditional information density}
\label{s:stronga}

\begin{definition}[Conditional information density]
For an arbitrary source-side information pair $(\Xp,\Yp)$, the
{\em conditional information density} of blocklength~$n$
is the random variable: 
$-\log P(X_1^n|Y_1^n)=-\log P_{X_1^n|Y_1^n}(X_1^n|Y_1^n)$.
\end{definition}

\noindent
[Throughout the paper, `$\log$' denotes `$\log_2$', the logarithm taken
to base~2, and all familiar information theoretic quantities
are expressed in bits.]

The starting point
is the following almost sure (a.s.) approximation result
between the description lengths
$\ell(f_n(X_1^n|Y_1^n))$ of an arbitrary 
sequence of compressors and the conditional
information density
$-\log{P(X_1^n|Y_1^n))}$ of an arbitrary 
source-side information pair $(\Xp,\Yp)$.
When it causes no confusion, we drop the 
subscripts for PMFs and conditional PMFs,
e.g., simply writing $P(x_1^n|y_1^n)$ for 
$P_{X_1^n|Y_1^n}(x_1^n|y_1^n)$ as in the definition
above.
Recall the definition of the optimal compressors
$\{f^*_n\}$ from Section~\ref{descriptionofanoptimal}.

\begin{theorem} 
For any source-side information pair 
\label{normpointred}
$(\Xp, \Yp)$,
and any sequence $\{B_n\}$
that grows faster than logarithmically,
i.e., such that $B_n/\log n\to \infty$
as $n\to\infty$, we have:
\begin{enumerate}
\item[$(a)$] For any sequence of compressors
with side information
$\seq{f_n}$:
\ben
\liminf_{n\rightarrow \infty}
{\frac{\ell(f_n(X_1^n|Y_1^n))-[-\log{P(X_1^n|Y_1^n)}]}{B_n}} 
\geq 0, \qquad \mbox{a.s.}
\een
\item[$(b)$]
The optimal compressors $\seq{f^*_n}$ achieve
the above bound with equality.
\end{enumerate}
\end{theorem}
\begin{proof}
Fix $\epsilon > 0$ arbitrary and 
let $\tau = \tau_n = \epsilon B_n$. Applying 
the general converse in 
\cite[Theorem~3.3]{gavalakis-arxiv:19}
with $X_1^n,Y_1^n$ in place of $X,Y$ and $\mathcal{X}^n,\mathcal{Y}^n$ 
in place of $\mathcal{X},\mathcal{Y}$, gives,
$$
\Pbig{\ell(f(X_1^n|Y_1^n)) \leq -\log{P(X_1^n|Y_1^n)} - \epsilon B_n} 
\leq 2^{\log{n} - \epsilon B_n}(\lfloor\log{|\mathcal{X}|}\rfloor + 1),
$$
which is summable in $n$. Therefore, by the Borel-Cantelli
lemma we have that, eventually, almost surely,
$$
\ell(f(X_1^n|Y_1^n)) + \log{P(X_1^n|Y_1^n)} > - \epsilon B_n.
$$
Since $\epsilon>0$ was arbitrary, this implies~$(a)$.
Part~$(b)$ follows from~$(a)$ together with the fact
that, with probability~1,
$\ell(f^*_n(X_1^n|Y_1^n))+\log{P(X_1^n|Y_1^n)} \leq 0$,
by the general achievability 
result in \cite[Theorem~3.1]{gavalakis-arxiv:19}.
\end{proof}

\subsection{First-order asymptotics}
\label{s:first}

For any source-side information pair $(\Xp,\Yp)$, the 
{\em conditional entropy rate} $H(\Xp|\Yp)$ is defined
as:
$$
H(\Xp|\Yp)=\limsup_{n\to\infty}\frac{1}{n}H(X_1^n|Y_1^n).$$
Throughout $H(Z)$ and $H(Z|W)$ denote the discrete entropy
of $Z$ and the conditional entropy of $Z$ given $W$,
in bits.
If $(\Xp,\Yp)$ are jointly stationary, then the above 
$\limsup$ is in fact a limit,
and it is equal to $H(\Xp,\Yp)-H(\Yp)$,
where $H(\Xp,\Yp)$ and $H(\Yp)$
are the entropy rates of $(\Xp,\Yp)$ and
of $\Yp$, respectively \cite{cover:book2}.
Moreover, if $(\Xp,\Yp)$
are also jointly ergodic, then by applying
the Shannon-McMillan-Breiman theorem
\cite{cover:book2} to $\Yp$ and to the pair $(\Xp,\Yp)$,
we obtain its conditional version: 
\begin{equation} \label{conditionalmcmillan}
-\frac{1}{n}\log{P(X_1^n|Y_1^n)} \rightarrow H(\Xp|\Yp),
\qquad \mbox{a.s.}
\end{equation}
The next result states that the conditional entropy rate
is the best asymptotically achievable compression rate,
not only in expectation but also with probability~1.
It is a consequence of Theorem~\ref{normpointred}
with $B_n=n$, combined with~(\ref{conditionalmcmillan}).

\begin{theorem}
\label{t:kieffer}
Suppose $(\Xp,\Yp)$ is a jointly stationary and ergodic
source-side information pair with conditional 
entropy rate $H(\Xp|\Yp)$.
\begin{enumerate}
\item[$(a)$] For any sequence 
of compressors with side information
$\seq{f_n}$:
\ben
\liminf_{n\rightarrow \infty}{\frac{\ell(f_n(X_1^n|Y_1^n))}{n}} 
\geq H(\Xp|\Yp), \qquad \mbox{a.s.}
\een
\item[$(b)$] The optimal compressors $\seq{f^*_n}$ achieve
the above bound with equality.
\end{enumerate}
\end{theorem}

\subsection{Finer asymptotics}
\label{s:finer}

The refinements of Theorem~\ref{t:kieffer} presented in this
section will be derived as consequences of the general
approximation results in Theorem~\ref{normpointred},
combined with corresponding refined asymptotics
for the conditional information density $-\log P(X_1^n|Y_1^n)$.
For clarity of exposition these are stated separately, in
Section~\ref{s:density} below.

The results of this section will be established for a
class of jointly stationary and ergodic source-side information 
pairs $(\Xp,\Yp)$, that includes all Markov chains with positive
transition probabilities. The relevant conditions, 
in their most general form, will be given in terms 
of the following mixing coefficients.

\begin{definition}
\label{d:mix}
Suppose $\Zp=\{Z_n\;;\;n\in\IN\}$ is a stationary process on a finite 
alphabet $\clZ$. For $-\infty \leq i \leq j \leq \infty$, 
let $\mathcal{F}_{i}^{j}$ 
denote the $\sigma$-algebra generated by $Z_{i}^{j}$.  
For $d \geq 1$, define:
\begin{align*}
\alpha^{(\Zps)}(d) &= 
\sup\bigl\{|\PP(A \cap B) - \PP(A)\PP(B)|\;;\;
A\in \mathcal{F}_{-\infty}^{0}, B \in \mathcal{F}_{d}^{\infty}\bigr\},\\
\gamma^{(\Zps)}(d) &= 
\max_{z \in \clZ}
{\Ebig{ \bigm|\log{\PP(Z_{0} = z|Z_{-\infty}^{-1})} 
- \log{\PP(Z_{0} = z|Z_{-d}^{-1})}\bigm|} }.
\end{align*}
\end{definition}
Note that if $\Zp$ is an ergodic Markov 
chain of order $k$, then
$\alpha^{(\Zps)}(d)$ decays exponentially 
fast \cite{bradley:86},
and $\gamma^{(\Zps)}(d)=0$ for all $d\geq k$.
Moreover, if $(\Xp,\Yp)$ is a Markov chain
with all positive transition probabilities,
then $\gamma^{(\Yps)}(d)$ also decays
exponentially fast; cf.~\cite[Lemma~2.1]{han:11}.

Throughout this section we will assume that 
the following conditions hold:

\medskip

\noindent
{\bf Assumption~(M). }
The source-side information pair $(\Xp,\Yp)$ 
is stationary and satisfies one
of the following three conditions:
\begin{itemize}
\item[$(a)$] $(\Xp,\Yp)$ is a Markov chain with all positive transition
probabilities; or
\item[$(b)$] $(\Xp,\Yp)$ as well as $\Yp$ are $k$th order,
irreducible and aperiodic Markov chains; or
\item[$(c)$]
$(\Xp,\Yp)$ is jointly ergodic and
satisfies the following mixing conditions:\footnote{Our source-side 
	information pairs $(\Xp,\Yp)$ are
	only defined for $(X_n,Y_n)$ with $n\geq 1$,
	whereas the coefficients $\alpha^{(\Zpss)}(d)$
	and $\gamma^{(\Zpss)}(d)$ are defined for two-sided
	sequences $\{Z_n\;;\;n\in\IN\}$. But this 
	does not impose an additional restriction, since
	any one-sided stationary process can be extended
	to a two-sided one by the Kolmogorov 
	extension theorem~\cite{bl:pm}.}%
\be
\alpha^{(\Xps,\Yps)}(d) = O(d^{-336}),\qquad
\gamma^{(\Xps,\Yps)}(d) = O(d^{-48}),\qquad\mbox{and}\qquad
\gamma^{(\Yps)}(d) = O(d^{-48}).
	\label{eq:mixing}
\ee
\end{itemize}

In view of the discussion following Definition~\ref{d:mix},
$(a)\Rightarrow(c)$ and $(b)\Rightarrow(c)$. Therefore, all results
stated under assumption~(M) will be proved under the 
weakest set of conditions, namely, 
that~(\ref{eq:mixing}) hold.

\begin{definition}\label{asymptoticvariance}
For a source-side information pair $(\Xp,\Yp)$, 
the {\em conditional varentropy rate} is:
\begin{equation}
\label{variancedef}
\sigma^2(\Xp|\Yp) 
= \limsup_{n \rightarrow \infty}{\frac{1}{n}
\VAR{\left(-\log{P(X_1^n|Y_1^n)}\right)}}.
\end{equation}
\end{definition}


\noindent
Under the above assumptions, the $\limsup$ in~(\ref{variancedef})
is in fact a limit. Lemma~\ref{lem:varentropy} is proved in the 
Appendix.

\begin{lemma}
\label{lem:varentropy}
Under assumption~{\em (M)},
the conditional varentropy rate $\sigma^2(\Xp|\Yp)$ is:
\begin{equation*}
\sigma^2(\Xp|\Yp) = \lim_{n \rightarrow \infty}{\frac{1}{n}
\VAR{\left(-\log{P(X_1^n|Y_1^n)}\right)}} = 
\lim_{n \rightarrow \infty}{\frac{1}{n}
\VAR{\Biggl(-\log\Big(
\frac{P(X_1^n,Y_1^n|X_{-\infty}^0,Y_{-\infty}^0)}
{P(Y_1^n|Y_{-\infty}^0)}\Big)
\Biggr)}}.
\end{equation*}
\end{lemma}

Our first main result of this section 
is a ``one-sided'' central limit theorem (CLT), which
states
that the description lengths $\ell(f_n(X_1^n|Y_1^n))$ of
an arbitrary sequence of compressors
with side information, $\{f_n\}$, are asymptotically
at best Gaussian, with variance $\sigma^2(\Xp|\Yp)$.
Recall the optimal compressors $\{f^*_n\}$ described in 
Section~\ref{descriptionofanoptimal}

\begin{theorem}[CLT for codelengths]
\label{t:cltc}
Suppose $(\Xp,\Yp)$ satisfy assumption~{\em (M)},
and let $\sigma^2=\sigma^2(\Xp|\Yp)>0$ denote the 
conditional varentropy rate~{\em (\ref{variancedef})}.
Then there exists a sequence of random variables $\{Z_n\;;\;n\geq 1\}$
such that:
\begin{enumerate}
\item[$(a)$] For any sequence of compressors with side information,
$\seq{f_n}$, we have,
\begin{equation}
\liminf_{n\rightarrow \infty}{\Biggl[\frac{\ell(f_n(X_1^n|Y_1^n)) - H(X_1^n|Y_1^n)}{\sqrt{n}} - Z_n\Biggr]} \geq 0, \qquad \mbox{a.s.},
\label{eq:cltLB}
\end{equation}
where,
$Z_n \rightarrow N(0,\sigma^2),$
in distribution, as $n\to\infty.$
\item[$(b)$] The optimal compressors $\{f^*_n\}$ 
achieve the lower bound in~{\em (\ref{eq:cltLB})} with equality.
\end{enumerate}
\end{theorem}


\begin{proof}
Letting $Z_n=[-\log P(X_1^n|Y_1^n)]/\sqrt{n}$, $n\geq 1$, 
and taking 
$B_n = \sqrt{n}$, both results follow by combining 
the approximation results 
of Theorem~\ref{normpointred} 
with the corresponding CLT for the conditional information
density in Theorem~\ref{cltlogp}.
\end{proof}

Our next result is in the form of a 
``one-sided'' law of the iterated logarithm
(LIL) which states that, with probability~1,
the description lengths of any compressor
with side information will have inevitable
fluctuations of order  $\sqrt{2\sigma^2n\log_e\log_2 n}$ bits
around the conditional entropy rate $H(\Xp|\Yp)$;
throughout, $\log_e$ denotes the natural logarithm to base $e$.

\begin{theorem}[LIL for codelengths]
\label{t:lilc}
Suppose $(\Xp,\Yp)$ satisfy assumption~{\em (M)},
and let $\sigma^2=\sigma^2(\Xp|\Yp)>0$ denote the 
conditional varentropy rate~{\em (\ref{variancedef})}.
Then:
\begin{enumerate}[label=(\alph*)]
\item[$(a)$] For any sequence of compressors with
side information, $\seq{f_n}$, we have:
 \begin{align}
& 
\limsup_{n \rightarrow \infty}{\frac{\ell(f_n(X_1^n|Y_1^n)) 
- H(X_1^n|Y_1^n)}{\sqrt{2n\log_e{\log_e{n}}}} } \geq \sigma, 
\qquad \mbox{a.s.},
	\label{eq:lilLB1}\\
\mbox{and}\qquad
&
\liminf_{n \rightarrow \infty}{\frac{\ell(f_n(X_1^n|Y_1^n)
)- H(X_1^n|Y_1^n)}{\sqrt{2n\log_e{\log_e{n}}}} } \geq -\sigma, \qquad \mbox{a.s.}
	\label{eq:lilLB2}
\end{align}
\item[$(b)$] The optimal compressors $\{f^*_n\}$ 
achieve the lower bounds in~{\em (\ref{eq:lilLB1})} 
and~{\em (\ref{eq:lilLB2})} with equality.
\end{enumerate}
\end{theorem}

\vspace*{-0.2in}

\begin{proof}
Taking 
$B_n = \sqrt{2n\log_2\log_e n}$, the results of
the theorem again follow by combining 
the approximation results 
of Theorem~\ref{normpointred} 
with the corresponding LIL for the conditional information
density in Theorem~\ref{LILlogpth}.
\end{proof}


\medskip

\noindent
{\bf Remarks. }
\begin{enumerate}
\item
Although the results in Theorems~\ref{t:cltc}
and~\ref{t:lilc} are stated for one-to-one
compressors $\{f_n\}$, they remain valid
for the class of prefix-free compressors.
Since prefix-free codes are certainly one-to-one,
the converse bounds in Theorem~\ref{t:cltc}~$(a)$
and~\ref{t:lilc}~$(a)$ are valid as stated,
while for the achievability results it suffices
to consider compressors $f_n^p$ with description lengths
$\ell(f_n^p(x_1^n|y_1^n)))=\lceil-\log P(x_1^n|y_1^n)\rceil$,
and then apply Theorem~\ref{cltlogp}.
\item
Theorem~\ref{t:cltc} says that the compression rate of any sequence
of compressors $\{f_n\}$ will have at best Gaussian 
fluctuations around $H(\Xp|\Yp)$,
$$\frac{1}{n}\ell(f^*_n(X_1^n|Y_1^n))\approx 
N\Big(H(\Xp|\Yp),\frac{\sigma^2(\Xp|\Yp)}{n}\Big),
\qquad\mbox{bits/symbol},$$
and similarly Theorem~\ref{t:lilc} says that,
with probability~1, the description lengths
will have inevitable fluctuations of 
approximately $\pm\sqrt{2n\sigma^2\log_e\log_e n}$ bits
around $nH(\Xp|\Yp)$.

As both of these vanish when $\sigma^2(\Xp|\Yp)$ is
zero, we note that, if the source-side information
pair $(\Xp,\Yp)$ is memoryless, so that $\{(X_n,Y_n)\}$
are independent and identically distributed, 
then the conditional varentropy rate reduces to,
$$\sigma^2(\Xp|\Yp)=\VAR(-\log P(X_1|Y_1)),$$
which is equal to zero if and only if, for each $y\in\clY$,
the conditional distribution of $X_1$ given $Y_1=y$
is uniform on a subset $\clX_y\subset\clX$, where
all the $\clX_y$ have the same cardinality.

In the more general case when both the pair process
$(\Xp,\Yp)$ and the side information $\Yp$
are Markov chains, necessary and sufficient conditions
for $\sigma^2(\Xp|\Yp)$ to be zero were recently
established in \cite{gavalakis-arxiv:19}.

\item
In analogy with the source dispersion
for the problem of lossless compression
without side information \cite{kontoyiannis-verdu:14,tan:14},
for an arbitrary source-side information pair $(\Xp,\Yp)$
the conditional dispersion $D(\Xp|\Yp)$ was recently
defined \cite{gavalakis-arxiv:19} as,
$$D(\Xp|\Yp) 
= \limsup_{n\to\infty}{\frac{1}{n}
{\VAR\bigl[{\ell(f_n^*(X_1^n|Y_1^n))\bigr]}}}.
$$
There, it was shown that when
both the pair $(\Xp,\Yp)$ and 
$\Yp$ itself are irreducible and aperiodic Markov chains,
the conditional dispersion coincides with
the conditional varentropy rate:
$$D(\Xp|\Yp)\! = 
\lim_{n\to\infty}{\frac{1}{n}{\VAR\bigl[{\ell(f_n^*(X_1^n|Y_1^n))\bigr]}}} 
= \sigma^2(\Xp|\Yp)
< \infty.
$$
\end{enumerate}

\subsection{Asymptotics of the conditional information density}
\label{s:density}

Here we show that the conditional information density itself,
$-\log P(X_1^n|Y_1^n)$, satisfies a CLT and a LIL.
The next two theorems are consequences of the
almost sure invariance principle established in
Theorem~\ref{philippstout}, 
in the Appendix.

\begin{theorem}
[CLT for the conditional information density]
\label{cltlogp}
$\;$ Suppose $(\Xp,\Yp)$ satisfy assumption~{\em (M)},
and let $\sigma^2=\sigma^2(\Xp|\Yp)>0$ denote the 
conditional varentropy rate~{\em (\ref{variancedef})}.
Then, as $n\to\infty$:
\begin{equation} \label{cltstatementlogp}
\frac{-\log{P(X_1^n|Y_1^n)} 
- H(X_1^n|Y_1^n)}{\sqrt{n}} \rightarrow N(0,\sigma^2),
\qquad\mbox{in distribution}.
\end{equation}
\end{theorem}

\begin{proof}
The conditions~(\ref{eq:mixing}),
imply that, as $n\to\infty$,
$[nH(\Xp,\Yp) - H(X_1^n,Y_1^n)]/\sqrt{n} \rightarrow 0$,
and 
$[nH(\Yp) - H(Y_1^n)]/\sqrt{n} \rightarrow 0$,
cf.~\cite{philipp-stout:book}, therefore also,
$[nH(\Xp|\Yp) - H(X_1^n|Y_1^n)]/\sqrt{n} \rightarrow 0$,
so it suffices to show that, as $n\to\infty$,
 \begin{equation}\label{cltstatmentlogpsuff}
  \frac{-\log{P(X_1^n|Y_1^n)} - nH(\Xp|\Yp)}{\sqrt{n}} 
  \rightarrow N(0,\sigma^2).
	\qquad\mbox{in distribution}.
 \end{equation}
Let $D=D([0,1],\mathbb{R})$ denote the space of \textit{cadlag} 
(right-continuous with left-hand limits) functions
from $[0,1]$ to $\mathbb{R}$, and define, for each $t\geq 0$,
$S(t) = \log{P(X_{1}^{\lfloor t \rfloor}|Y_{1}^{\lfloor t \rfloor})}
+ t H(\Xp|\Yp)$,
as in Theorem~\ref{philippstout} in the Appendix.
For all $n\geq1, t\in[0,1]$, define
$S_n(t) = S(nt)$. Then Theorem~\ref{philippstout} 
implies that, as $n\to\infty$,
$$\Big\{\frac{1}{\sigma\sqrt{n}}S_n(t)\;;\;t \in [0,1]\Big\}
\rightarrow \{B(t)\;;\; t \in [0,1]\}, \qquad \text{
weakly in $D$},
$$
where $\{B(t)\}$ is a standard Brownian motion;
see, e.g., \cite[Theorem~E, p.~4]{philipp-stout:book}.
In particular, this implies that,
$$\frac{1}{\sigma\sqrt{n}}S_n(1)
\rightarrow B(1)\sim N(0,1),\qquad\mbox{in distribution},
$$
which is exactly~(\ref{cltstatmentlogpsuff}).
\end{proof}
 
\begin{theorem}
[LIL for the conditional information density]
\label{LILlogpth}
$\;$
Suppose $(\Xp,\Yp)$ satisfy assumption~{\em (M)},
and let $\sigma^2=\sigma^2(\Xp|\Yp)>0$ denote the 
conditional varentropy rate~{\em (\ref{variancedef})}.
Then:
\begin{align} \label{LILlogp}
 &\limsup_{n \rightarrow \infty}{\frac{-\log{P(X_1^n|Y_1^n)} - H(X_1^n|Y_1^n)}{\sqrt{2n\log_e{\log_e{n}}}} } = \sigma, \qquad \mbox{a.s.},\\
\mbox{and}\qquad
 &\liminf_{n \rightarrow \infty}{\frac{-\log{P(X_1^n|Y_1^n)} - H(X_1^n|Y_1^n)}{\sqrt{2n\log_e{\log_e{n}}}} } = -\sigma, \qquad \mbox{a.s.}
\label{LILlogp2}
 \end{align}
\end{theorem}

\begin{proof}
As in the proof of (\ref{cltstatementlogp}), it suffices to 
prove (\ref{LILlogp}) with $nH(\Xp|\Yp)$ 
in place of $H(X_1^n|Y_1^n)$. But this is immediate 
from Theorem~\ref{philippstout}, since, for a standard
Brownian motion $\{B(t)\}$,
$$  \limsup_{t \rightarrow \infty}{\frac{B(t)}
{\sqrt{2t\log_e{\log_e{t}}}} } = 1,
 \qquad \mbox{a.s.},
$$
see, e.g., \cite[Theorem~11.18]{Kal:Foundations:2002}.
And similarly for~(\ref{LILlogp2}).
\end{proof}

\section{Idealised LZ Compression with Side Information}
\label{lempelzivsection}

Consider the following idealised version of 
Lempel-Ziv-like
compression with side information.
For a given source-side information pair
$(\Xp,\Yp)=\{(X_n,Y_n)\;;\;n\in\IN\}$, the encoder
and decoder both have access to the infinite past
$(X_{-\infty}^0,Y_{-\infty}^0)$ and to the
current side information $Y_1^n$.
The encoder describes $X_1^n$ to the decoder as follows. 
First she searches for the first appearance of 
$(X_1^n,Y_1^n)$ in 
the past $(X_{-\infty}^0,Y_{-\infty}^0)$, 
that is, for the first $r\geq 1$ such that,
$(X_{-r+1}^{-r+n},Y_{-r+1}^{-r+n})=(X_1^n,Y_1^n)$.
Then she counts how many times $Y_1^n$ appears in
$Y_{-\infty}^0$ between
locations $-r+1$ and~$0$, namely, how many indices $1\leq j<r$
there are, such that $Y_{-j+1}^{-j+n}=Y_1^n$. Say there
are $(\clR_n-1)$ such $j$s. She describes $X_1^n$ to the
decoder by telling him to look at the $\clR_n$th position
where $Y_1^n$ appears in the past $Y_{-\infty}^0$,
and read off the corresponding $X$ string.

This description takes $\approx\log \clR_n$ bits,
and, as it turns out, the resulting compression
rate is asymptotically optimal: As $n\to\infty$,
with probability~1,
\be
\frac{1}{n}\log \clR_n\to H(\Xp|\Yp),\qquad\mbox{bits/symbol}.
\label{eq:heur}
\ee
Moreover, it is second-order optimal,
in that it achieves equality in 
the CLT and LIL bounds given 
in Theorems~\ref{t:cltc}
and~\ref{t:lilc} of Section~\ref{pointwiseasymptotics}.

Our purpose in this section is to make these statements
precise. We will prove~(\ref{eq:heur})
as well as its CLT and LIL refinements,
generalising the corresponding results
for recurrence times without side information
in~\cite{kontoyiannis-jtp}.

The use of recurrence times in understanding
the Lempel-Ziv (LZ) family of algorithms
was introduced by Willems \cite{willems:1}
and Wyner and Ziv \cite{wyner-ziv:1,wyner-ziv:2}.
In terms of practical methods for compression
with side information,
Subrahmanya and Berger \cite{subrahmanya:95} 
proposed a side information analog of the sliding 
window LZ algorithm \cite{ziv-lempel:1}, 
and Uyematsu and Kuzuoka \cite{uyematsu:03}
proposed a side information version of
the incremental parsing 
LZ algorithm \cite{ziv-lempel:2}.
The Subrahmanya-Berger algorithm
was shown to be asymptotically optimal 
in \cite{jacob:08} and \cite{jain-bansal}. 
Different types of LZ-like 
algorithms for compression with side information 
were also considered in \cite{tock:05}.

Throughout this section, we assume $(\Xp,\Yp)$ is
a jointly stationary and ergodic source-side
information pair, with values in the finite
alphabets $\clX,\clY$, respectively. 
We use bold lower-case letters $\xp,\yp$ without subscripts
to denote infinite realizations $x_{-\infty}^\infty,
y_{-\infty}^\infty$ of $\Xp,\Yp$,
and the corresponding bold capital letters $\Xp,\Yp$ 
without subscripts to denote 
the entire process, $\Xp=X_{-\infty}^\infty,
\Yp=Y_{-\infty}^\infty$.

The main quantities of interest
are the recurrence times defined next.

\begin{definition}[Recurrence times]
For a realization $\xp$ of the process $\Xp$,
and $n\geq 1$,
define the {\em repeated recurrence times} $\clR_n^{(j)}(\xp)$ of $x_1^n$,
recursively, as:
\begin{align*}
\clR_{n}^{(1)}(\xp) &=  \inf \{i \geq 1: x_{-i+1}^{-i+n} = x_{1}^{n} \}, \\
\clR_{n}^{(j)}(\xp) &=  \inf \{i > \clR_{n}^{(j-1)}(x) : x_{-i+1}^{-i+n} = x_{1}^{n}\},
\qquad j > 1.  
\end{align*}
For a realization $(\xp,\yp)$ of the pair $(\Xp,\Yp)$ and $n\geq 1$,
the {\em joint recurrence time} $\clR_n(\xp,\yp)$ of $(x_1^n,y_1^n)$ 
is defined as,
\begin{equation*}
\clR_{n}(\xp,\yp) =  \inf \{i \geq 1: (x,y)_{-i+1}^{-i+n} = (x,y)_{1}^{n} \},
\end{equation*}
and the {\em conditional recurrence time} $\clR_n(\xp|\yp)$ of $x_1^n$
among the appearances $y_1^n$ is:
\begin{equation*}
\clR_{n}(\xp|\yp) =  
\inf \Big\{i \geq 1: x_{-\clR_{n}^{(i)}(y)+1}^{-\clR_{n}^{(i)}(y)+n} 
= x_{1}^{n} \Big\}.
\end{equation*}
\end{definition}


An important tool in the asymptotic analysis
of recurrence times is Kac's Theorem \cite{kac:47}. 
Its conditional version in Theorem~\ref{kac}
was first established in \cite{jacob:08}
using Kakutani's induced transformation
\cite{kakutani:43,shields:book}.

\begin{theorem}[Conditional Kac's theorem]
\label{kac}
{\em \cite{jacob:08}}
Suppose $(\Xp,\Yp)$ is a
jointly stationary and ergodic 
source-side information pair.
For any pair of strings
$x_{1}^{n}\in \mathcal{X}^n$,
$y_1^{n} \in \mathcal{Y}^n$:
\begin{equation*}
\BBE[\clR_{n}(\Xp|\Yp) | X_{1}^{n} = x_{1}^{n}, Y_{1}^{n} = y_{1}^{n}]
= \frac{1}{P(x_{1}^{n}|y_{1}^{n})}.
\end{equation*}
\end{theorem}

The following result states that we can asymptotically
approximate $\log \clR_n(X|Y)$ by the 
conditional information density
not just in expectation as in Kac's theorem,
but also with probability~1.
Its proof is in Appendix~\ref{recurrence}.

\begin{theorem} 
\label{strongapproxtheorem}
Suppose $(\Xp,\Yp)$ is a
jointly stationary and ergodic 
source-side information pair.
For any sequence $\{c_n\}$ of non-negative 
real numbers such that $\sum_n{n2^{-c_n}}<\infty$, we have:
\ben
(i)&&
\log \clR_{n}(\Xp|\Yp)
-\log\Big(\frac{1}{P(X_{1}^{n}|Y_{1}^{n})}\Big)
\leq c_n, \qquad \text{eventually a.s.}\\
(ii)&&
\log \clR_{n}(\Xp|\Yp)
-\log \Big(\frac{1}{P(X_{1}^{n}|Y_{1}^{n}, Y_{-\infty}^{0}, X_{-\infty}^{0})}
	\Big)
 \geq -c_n, \qquad \text{eventually a.s.}\\
(iii)&&
\log \clR_{n}(\Xp|\Yp)
-\log\Big(
\frac
{P(Y_{1}^{n}|Y_{-\infty}^{0})}
{P(X_{1}^{n},Y_{1}^{n}|Y_{-\infty}^{0}, X_{-\infty}^{0})}
\Big) \geq -2c_n, \qquad \text{eventually a.s.}
\een
\end{theorem}

Next we state the main consequences of Theorem~\ref{strongapproxtheorem}
that we will need. Recall the definition
of the coefficients $\gamma^{(\Zps)}(d)$ from 
Section~\ref{s:finer}.
Corollary~\ref{strongapproxcorollary}
is proved in Appendix~\ref{recurrence}.

\begin{corollary}
Suppose $(\Xp,\Yp)$ are jointly stationary and ergodic.
\label{strongapproxcorollary}
\begin{enumerate}
\item[$(a)$]
If, in addition,
$\sum_{d}{\gamma^{(\Xps,\Yps)}(d)} < \infty$ 
and  $\sum_{d}{\gamma^{(\Yps)}(d)} < \infty$,
then for any $\beta > 0$:
\begin{equation*}
\log[\clR_{n}(\Xp|\Yp)P(X_{1}^{n}|Y_{1}^{n})] = o(n^{\beta}), \qquad\mbox{a.s.}
\end{equation*}

\item[$(b)$]
In the general jointly ergodic case, we have:
\begin{equation*}
\log[\clR_{n}(\Xp|\Yp)P(X_{1}^{n}|Y_{1}^{n})] = o(n), \qquad\mbox{a.s.}
\end{equation*}
\end{enumerate}
\end{corollary}

From part~$(b)$ combined with the Shannon-McMillan-Breiman theorem
as in~(\ref{conditionalmcmillan}), we obtain the result~(\ref{eq:heur})
promised in the beginning of this section:
$$\lim_{n\to\infty}\frac{1}{n}\log \clR_n(\Xp|\Yp)\to H(\Xp|\Yp),
\qquad\mbox{a.s.}$$
This was first established in \cite{jacob:08}.
But at this point we have already done the work required
to obtain much finer asymptotic results for the
conditional recurrence time.

For any pair of infinite realizations $(\xp,\yp)$
of $(\Xp,\Yp)$, let $\{\clR^{(\xps|\yps)}(t)\;;\; t \geq 0\}$ be the 
continuous-time path, defined as:
\ben
\clR^{(\xps|\yps)}(t) &=& 0,\qquad\mbox{for}\;t < 1,\\
\clR^{(\xps|\yps)}(t) &=& \log{\clR_{\lfloor t \rfloor}(\xp|\yp)} - 
\lfloor t \rfloor H(\Xp|\Yp),
\qquad\mbox{for}\;t\geq 1.
\een
The following theorem is a direct consequence
of Corollary~\ref{strongapproxcorollary}~$(a)$
combined with Theorem~\ref{philippstout} in the Appendix.
Recall assumption~(M) from Section~\ref{s:finer}.

\begin{theorem}
\label{recurrenceinvariance}
Suppose $(\Xp,\Yp)$ satisfy assumption~{\em (M)},
and let
$\sigma^2=\sigma^2(\Xp|\Yp)>0$ denote the conditional varentropy rate.
Then $\{\clR^{(\Xps|\Yps)}(t)\}$ can be redefined on a richer probability 
space that contains a standard Brownian motion 
$\{B(t)\;;\; t \geq 0\}$ such that, for any $\lambda < 1/294$:
\begin{equation*}
\clR^{(\Xps|\Yps)}(t) - \sigma B(t) = O(t^{1/2-\lambda}), \qquad \mbox{a.s.}
\end{equation*}
\end{theorem}

Two immediate consequences of Theorem~\ref{recurrenceinvariance}
are the following: 

\begin{theorem}[CLT and LIL for the conditional recurrence times]
\label{recurrenceCLT}
$\;$ Suppose $(\Xp,\Yp)$ satisfy assumption~{\em (M)}
and let 
$\sigma^2=\sigma^2(\Xp|\Yp)>0$ denote the conditional varentropy rate.
Then:
\ben
(a)&&
\frac{\log{\clR_n(\Xp|\Yp)} - H(X_1^n|Y_1^n)}{\sqrt{n}} 
\rightarrow N(0,\sigma^2),\qquad 
\mbox{in distribution},\;
\mbox{as}\;n\to\infty.\\
(b)&&
 \limsup_{n \rightarrow \infty}{\frac{\log{\clR_n(\Xp|\Yp)} 
- H(X_1^n|Y_1^n)}{\sqrt{2n\log_e{\log_e{n}}}} } = \sigma, \qquad\mbox{a.s.}
\een
\end{theorem}


\begin{appendices}

\appendixpage

\section{Invariance Principle for the Conditional Information Density}
\label{philippstoutproof}

This Appendix is devoted to the proof of Theorem~\ref{philippstout},
which generalises the corresponding almost sure invariance principle of
Philipp and Stout \cite[Theorem~9.1]{philipp-stout:book}
for the (unconditional) information density $-\log P(X_1^n)$.

\begin{theorem} 
\label{philippstout}
Suppose $(\Xp,\Yp)$ is a jointly stationary and ergodic process,
satisfying the mixing 
conditions~{\em (\ref{eq:mixing})}.
For $t \geq 0$, let,
\begin{equation} 
\label{stdef}
S(t) = \log{P(X_{1}^{\lfloor t \rfloor}|Y_{1}^{\lfloor t \rfloor})} 
+ t H(\Xp|\Yp).
\end{equation}
Then the following series converges:
\begin{align*}
\sigma^{2} 
&= 	\mathbb{E}\bigl[\log{P(X_{0},Y_{0}|X_{-\infty}^{-1},
	Y_{-\infty}^{-1})}+H(\Xp|\Yp)\bigr]^{2} \\
&+ 
	2\sum_{k=1}^{\infty}{\mathbb{E}
	\Big\{\bigl[ \log{P(X_{0},Y_{0}
	|X_{-\infty}^{-1},Y_{-\infty}^{-1})} + H(\Xp|\Yp)\bigr]
	\bigl[\log{P(X_{k},Y_{k}|X_{-\infty}^{k-1},Y_{-\infty}^{k-1})} 
	+ H(\Xp|\Yp)\bigr]}\Big\}.
\end{align*}
If $\sigma^2 > 0$,
then, without changing its distribution, we can redefine 
the process $\{S(t)\;;\;t\geq0\}$ on a richer probability space that 
contains a standard Brownian motion 
$\{B(t)\;;\; t \geq 0\}$, such that,
\begin{equation} \label{StBtapprox}
S(t) - \sigma B(t) = O(t^{\frac{1}{2} - \lambda}), \qquad \mbox{a.s.},
\end{equation}
as $t\to\infty$,
for each $\lambda < 1/294$.
\end{theorem}

To simplify the notation, we write 
$h=H(\Xp|\Yp)$ and define,
\begin{equation} 
\label{fdef}
f_{j} = \log\left({\frac{P(X_{j},Y_{j}|X_{-\infty}^{j-1},Y_{-\infty}^{j-1})}
{P(Y_{j}|Y_{-\infty}^{j-1})}}\right),
\qquad j\geq 0,
\end{equation}
so that, for example, the variance $\sigma^2$ in
the theorem becomes,
\be
\sigma^{2} 
= 	\mathbb{E}[(f_{0} + h)^{2}]
	+ 2\sum_{k=1}^{\infty}{\mathbb{E}[(f_{0} + h)(f_{k} + h)]}.
\label{eq:newvar}
\ee

\begin{lemma}\label{fsufficient}
If $\sum_{d}{\gamma^{(\Xps,\Yps)}(d)} < \infty$ 
and  $\sum_{d}{\gamma^{(\Yps)}(d)} < \infty$ then,
as $n\to\infty$:
$$\sum_{k = 1}^{n}{f_{k}} - \log{P(X_{1}^{n}|Y_{1}^{n})} = O(1), \qquad 
\mbox{a.s.}
$$
\end{lemma}

\begin{proof}
Let,
$$g_{j} = \log\left(\frac{P(X_{j},Y_{j}|X_{1}^{j-1},Y_{1}^{j-1})}
{P(Y_{j}|Y_{1}^{j-1})}\right), \qquad j\geq 2,$$
and,
$$g_{1} = \log
\left(\frac{P(X_{1},Y_{1})}{P(Y_{1})}\right) = \log P(X_{1}|Y_{1}).$$
We have, for $k\geq 2$,
\begin{align*}
\mathbb{E}|f_{k}-g_{k}| 
\leq &
	\;\mathbb{E}|\log P(X_{k},Y_{k}|X_{-\infty}^{k-1},Y_{-\infty}^{k-1}) 
	- \log P(X_{k},Y_{k}|X_{1}^{k-1},Y_{1}^{k-1})|\\
&
	+ \mathbb{E}|\log P(Y_{k}|Y_{-\infty}^{k-1}) 
	- \log P(Y_{k}|Y_{1}^{k-1})| \\
\leq &
	\;\sum_{x,y}
	{\mathbb{E}\big|\log P(X_{k} = x,Y_{k} = y 
	|X_{-\infty}^{k-1},Y_{-\infty}^{k-1}) 
	- \log P(X_{k} = x,Y_{k} = y|X_{1}^{k-1},Y_{1}^{k-1})\big|} \\
&+ 
	\sum_{y}{\mathbb{E}\big|\log P(Y_{k} 
	= y|Y_{-\infty}^{k-1}) - \log P(Y_{k} = y|Y_{1}^{k-1})\big|}\\
\leq & 
	\;
	|\clX||\clY|\gamma^{(\Xps,\Yps)}(k-1) 
	+ |\clY|\gamma^{(\Yps)}(k-1).
\end{align*}
Therefore, $\sum_{k =1}^{\infty}{\mathbb{E}|f_{k} - g_{k}|} < \infty$,
and by the monotone convergence theorem we have,
$$\sum_{k=1}^{\infty}{|f_{k}-g_{k}|} < \infty, \qquad \mbox{a.s.}$$
Hence, as $n\to\infty$,
$$\left|\sum_{k = 1}^{n}{f_{k}} - \log{P(X_{1}^{n}|Y_{1}^{n})}\right| 
\leq \sum_{k=1}^{n}{|f_{k}-g_{k}|} = O(1), \qquad \mbox{a.s.},$$
as claimed.
\end{proof}

\medskip

The following bounds are established in the proof of
\cite[Theorem 9.1]{philipp-stout:book}:

\begin{lemma} \label{factsfromproof}
Suppose $\Zp=\{Z_n\;;\;n\in\IN\}$ is a stationary 
and ergodic process on 
a finite alphabet, with entropy rate $H(\Zp)$, and such that 
$\alpha^{(\Zps)}(d) = O(d^{-336})$
and 
$\gamma^{(\Zps)}(d)= O(d^{-48}),$ as $d\to\infty$.

Let $f^{(\Zps)}_{k} = \log{P(Z_{k}|Z_{-\infty}^{k-1})}$, $k\geq 0$,
and put $\eta^{(\Zps)}_{n} = f^{(\Zps)}_{n} + H(\Zp)$, $n\geq0$.
Then:
\begin{enumerate}

\item \label{9.1} For each $r > 0$,
$\mathbb{E}\big[\big|f_{0}^{(\Zps)}\big|^{r}\big] < \infty.$

\item \label{9.2} For each $r \geq 2$ and $\epsilon > 0$,
$$\mathbb{E}\big[\big |f_{0}^{(\Zps)} - \log{P(Z_{0}|Z_{-k}^{-1})}\big|^{r}\big]
\leq
C(r,\epsilon)(\gamma^{(\Zps)}(k))^{\frac{1}{2}-\epsilon},$$
where $C(r,\epsilon)$ is a constant depending only on r and $\epsilon$. 

\item \label{norm4bound} For a constant $C > 0$ independent
of $n$, $\|\eta^{(\Zps)}_{n}\|_{4} \leq C.$

\item \label{approx} Let $\eta^{(\Zps)}_{n\ell} 
= \mathbb{E}[\eta^{(\Zps)}_{n}|\mathcal{F}_{n-\ell}^{n}]$. Then,
as $\ell\to\infty$:
$$
\|\eta^{(\Zps)}_{n} - \eta^{(\Zps)}_{n\ell}\|_{4} = O(\ell^{-11/2}).
$$
\end{enumerate}
\end{lemma}

Note that, under the assumptions of Theorem~\ref{philippstout}, the 
conclusions of Lemma~\ref{factsfromproof} apply to $\Yp$ as well as 
to the pair process $(\Xp,\Yp)$.

\begin{lemma} \label{lemma1}
For each $r>0$, we have,
$\mathbb{E}[|f_{0}|^{r}] < \infty.$
\end{lemma}

\begin{proof}
Simple algebra shows that,
$$
f_{0} = f_{0}^{(\Xps,\Yps)} - f_{0}^{(\Yps)}.$$
Therefore, by two applications of 
Lemma~\ref{factsfromproof}, part~\ref{9.1}, 
$$\|f_{0}\|_{r} \leq \|f_{0}^{(\Xps,\Yps)}\|_{r} 
+ \|f_{0}^{(\Yps)}\|_{r} < \infty.
$$
\end{proof}

The next bound follows
from Lemma~\ref{factsfromproof}, part~\ref{9.2}, 
upon applying the Minkowski inequality. 

\begin{lemma} \label{lemma2}
For each $r \geq 2$ and each $\epsilon > 0$,
$$\left\|f_{0} - \log\left(\frac{P(X_{0},Y_{0}|
X_{-k}^{-1},Y_{-k}^{-1})}{P(Y_{0}|Y_{-k}^{-1})}\right)\right\|_{r} 
\leq C_{1}(r,\epsilon)
[\gamma^{(\Xps,\Yps)}(k)]^{\frac{1-2\epsilon}{2r}} 
+ C_{2}(r,\epsilon)[\gamma^{(\Yps)}(k)]^{\frac{1-2\epsilon}{2r}}.
$$
\end{lemma}

\begin{lemma} 
\label{sigmawellfinite}
As $N\to\infty$:
$$
\mathbb{E}\left\{\left[\sum_{k\leq N}{(f_{k}+h)}\right]^2\right\} 
= \sigma^{2}N + O(1).
$$
\end{lemma}

\begin{proof}
First we examine the definition of the variance $\sigma^2$.
The first term in~(\ref{eq:newvar}),
$$\|f_0+h\|_2^2\leq (\|f_0\|_2+h)^2<\infty,$$
is finite by Lemma~\ref{lemma1}.
For the series in~(\ref{eq:newvar}), let, for $k\geq 0$, 
$$
\phi _{k} = 
\log
\left(\frac{P(X_{k},Y_{k}|X_{\lfloor k/2\rfloor}^{k-1},
Y_{\lfloor k/2\rfloor}^{k-1})}{P(Y_{k}
|Y_{\lfloor k/2\rfloor}^{k-1})}\right),
$$
and write,
\be
\mathbb{E}(f_{0} + h)(f_{k} + h) = 
\mathbb{E}(f_{0}+h)(f_{k}-\phi_{k}) + \mathbb{E}(f_{0}+h)(\phi_{k} + h).
\label{eq:triangle}
\ee
For the first term in the right-hand side above, we can bound, 
for any $\epsilon>0$,
\begin{align*}
|\mathbb{E}(f_{0}+h)(f_{k}-\phi_{k})|
&\leqa \|f_{0}+h\|_{2}\|f_{k}-\phi_{k}\|_{2} \\
&\leq [\|f_{0}\|_2+h]\|f_{k}-\phi_{k}\|_{2} \\
& \leqb 
	A C_1(2,\epsilon)\Big[
	\gamma^{(\Xps,\Yps)}(\lfloor k/2\rfloor)
	\Big]^{\frac{1}{4}- \frac{1}{2}\epsilon} 
	+A C_2(2,\epsilon)\Big[
	\gamma^{(\Yps)}(\lfloor k/2\rfloor)
	\Big]^{\frac{1}{4}- \frac{1}{2}\epsilon},
\end{align*}
where $(a)$ follows by the Cauchy-Schwarz inequality, 
and $(b)$ follows by Lemmas~\ref{lemma1}, and~\ref{lemma2},
with $A=\|f_{0}\|_2+h<\infty$. Therefore, taking
$\epsilon>0$ small enough and using the assumptions
of Theorem~\ref{philippstout},
\be
|\mathbb{E}(f_{0}+h)(f_{k}-\phi_{k})|
= O(k^{-12+24\epsilon}) = O(k^{-3}), \qquad
\mbox{as}\;k\to\infty.
\label{eq:term1}
\ee
For the second term in~(\ref{eq:triangle}),
we have that, for any $r>0$,
$\|\phi_k\|_r<\infty$,
uniformly over $k\geq 1$ by stationarity.
Also, since $f_{0}, \phi_{k}$ are measurable with respect 
to the $\sigma$-algebras generated by $(X_{-\infty}^{0},Y_{-\infty}^{0})$ 
and $(X_{\lfloor k/2\rfloor}^{\infty},
Y_{\lfloor k/2\rfloor}^{\infty})$, respectively, we can
apply \cite[Lemma~7.2.1]{philipp-stout:book} with $p=r=s=3$,
to obtain that, 
$$|\mathbb{E}(f_{0} + h)(\phi_{k} + h)| \leq
10 \|f_{0}+h\|_{3}\|\phi_{k} + h\|_{3}\alpha(\lfloor k/2\rfloor)^{1/3},
$$
where $\alpha(k) = \alpha^{(\Xps,\Yps)}(k) 
= O(k^{-48})$, as $k\to\infty$, by assumption. 
Therefore, a fortiori,
$$\mathbb{E}(f_{0} + h)(f_{k} + h) = O(k^{-3}),$$
and combining this with~(\ref{eq:term1})
and substituting in~(\ref{eq:triangle}),
implies that $\sigma^{2}$ in~(\ref{eq:newvar})
is well defined and finite. 

Finally, we have that, as $N\to\infty$,
\begin{align*}
\mathbb{E}\left\{\left[
\sum_{k\leq N}{(f_{k}+h)}\right]^2\right\}
&= N\mathbb{E}(f_{0} + h)^{2} + 2\sum_{k=0}^{N-1}{(N-k)\mathbb{E}(f_{0}+h)(f_{k}+h)} \\
&= N\sigma^{2} - 2\sum_{k=1}^{N-1}{k\mathbb{E}(f_{0}+h)(f_{k}+h)} - 2N\sum_{k=N}^{\infty}{\mathbb{E}(f_{0}+h)(f_{k}+h)} \\&
= \sigma^{2}N + O(1),
\end{align*}
as required.
\end{proof}

\newpage


\noindent{\bf Proof of Lemma~\ref{lem:varentropy}. }
Lemma~\ref{sigmawellfinite} states that the limit,
\begin{equation}
\label{eq:varA}
\lim_{n \rightarrow \infty}{\frac{1}{n}
\VAR{\left(-\log\Big(\frac{P(X_1^n,Y_1^n|X_{-\infty}^0,Y_{-\infty}^0)}
{P(Y_1^n|Y_{-\infty}^0)}\Big)\right)}}.
\end{equation}
exists and is finite.
Moreover, by Lemma~\ref{lemma2},
after an application of the Cauchy-Schwarz inequality,
we have that, as $n\to\infty$,
\begin{equation*}
\mathbb{E}
\left\{\left[
\sum_{k\leq n}
\left| 
\log\left(\frac{P(X_{k},Y_{k}|X_1^{k-1},Y_1^{k-1})}{P(Y_{k}|Y_1^{k-1})}\right)
- \log\left(\frac{P(X_k,Y_k|X_{-\infty}^{k-1},Y_{-\infty}^{k-1})}
{P(Y_k|Y_{-\infty}^{k-1})}\right)
\right|\right]^2\right\} = O(1),
\end{equation*}
therefore,
\begin{equation*}
\frac{1}{n}\Biggl\{\VAR{\Biggl(-\log{P(X_1^n|Y_1^n)}\Biggr)}
-\VAR{\Biggl(-\log\Big(\frac{P(X_1^n,Y_1^n|X_{-\infty}^0,Y_{-\infty}^0)}
{P(Y_1^n|Y_{-\infty}^0)}\Big)\Biggr)}\Biggr\} = o(1).
\end{equation*}
Combining this with~(\ref{eq:varA}) and the definition of $\sigma^2$,
completes the proof.
\qed

\medskip

\noindent
{\bf Proof of Theorem \ref{philippstout}. }
Note that we have already established the fact
that the expression for the variance 
converges to some $\sigma^2<\infty$.
Also, in view of Lemma \ref{fsufficient}, it is 
sufficient to prove the theorem for $\{\tilde{S}(t)\}$
instead of $\{S(t)\}$, where:
$$\tilde{S}(t) = \sum_{k \leq t}{(f_{k} + h)},
\qquad t\geq 0.
$$
This will be established by an application of
\cite[Theorem~7.1]{philipp-stout:book},
once we verify that conditions~(7.1.4), 
(7.1.5), (7.1.6), (7.1.7) and~(7.1.9) there
are all satisfied.

For each $n\geq 0$, let $\eta_{n} = f_{n} + h$,
where $f_n$ is defined in~(\ref{fdef}) and 
$h$ is the conditional entropy rate. 
First we observe that, by stationarity,
\begin{align} 
\mathbb{E}[\eta_{n}]
&= \mathbb{E}\left[\log\left(\frac{P(X_{n},Y_{n}|X_{-\infty}^{n-1},
	Y_{-\infty}^{n-1})}{P(Y_{n}|Y_{-\infty}^{n-1})}\right)\right]
	 + H(\Xp|\Yp)\nonumber\\ 
&= \mathbb{E}\big[\log{P(X_{0},Y_{0}|X_{-\infty}^{-1},Y_{-\infty}^{-1})}\big]
	 + H(\Xp,\Yp) - \mathbb{E}\big[\log{P(Y_{0}|Y_{-\infty}^{-1})}\big]
	 -H(\Yp)\nonumber\\ \label{7.1.4}
          &= 0,
\end{align}
where 
$H(\Xp,\Yp)$ and $H(\Yp)$ denote the entropy rates 
of $(\Xp,\Yp)$ and $\Yp$, respectively \cite{cover:book2}. 
Observe that, in the notation of Lemma~\ref{factsfromproof},
$\eta_{n} = \eta_{n}^{(\Xps,\Yps)} - \eta_{n}^{(\Yps)}$,
and $\eta_{n\ell} = \eta_{n\ell}^{(\Xps,\Yps)} 
- \eta_{n\ell}^{(\Yps)}.$
By Lemma~\ref{factsfromproof}, parts~\ref{norm4bound} and~\ref{approx}, 
there exist a constant $C$, independent of $n$ such that,
\begin{equation} \label{7.1.5}
\|\eta_{n}\|_{4} \leq C<\infty,
\end{equation}
and,
\begin{equation} \label{7.1.6}
\|\eta_{n} - \eta_{n\ell}\|_{4} = O(\ell^{-11/2}).
\end{equation}
And from Lemma~\ref{sigmawellfinite} we have,
\begin{equation} \label{7.1.7}
\mathbb{E}\left\{\Big(\sum_{n \leq N}{\frac{1}{\sigma}\eta_{n}} \Big)^{2}
\right\} = N + O(1).
\end{equation}

From~(\ref{7.1.4})--(\ref{7.1.7}) and the 
assumption that $\alpha^{(\Xps,\Yps)}(d) = O(d^{-336})$,
we have that all of the conditions~(7.1.4), 
(7.1.5), (7.1.6), (7.1.7) and (7.1.9) of 
\cite[Theorem~7.1]{philipp-stout:book} are satisfied for the 
random variables $\{\eta_{n}/\sigma\}$, with $\delta = 2$.
Therefore, $\{\tilde{S}(t)\;;\;t\geq 0\}$ can be redefined
on a possibly richer probability space, where there 
exists a standard Brownian motion $\{B(t)\;;\; t \geq 0\}$,
such that, as $t\to\infty$:
$$\frac{1}{\sigma}\tilde{S}(t) - B(t) = O(t^{1/2-\lambda}), \qquad \mbox{a.s.}$$
By Lemma \ref{fsufficient}, this completes the proof.
\qed

\section{Recurrence Times Proofs}
\label{recurrence}

In this appendix we provide the proofs of some of the more
technical results in Section~\ref{lempelzivsection}.
First we establish the following generalisation of
\cite[Lemma~16.8.3]{cover:book2}.

\begin{lemma}
\label{lemmacover}
Suppose $(\Xp,\Yp)$ is an arbitrary source-side information
pair. Then, for any sequence
$\{t_n\}$ of non-negative 
real numbers such that $\sum_n{2^{-t_n}}<\infty$,
we have:
\begin{equation*}
\log{\frac{P(Y_{1}^{n}|Y_{-\infty}^{0}, X_{-\infty}^{0})}
{P(Y_{1}^{n}|Y_{-\infty}^{0})}} \geq -t_n, \qquad \text{eventually a.s.}
\end{equation*}
\end{lemma}

\begin{proof}
Let $B(X_{-\infty}^{0},Y_{-\infty}^{0}) \subset \mathcal{Y}^{n}$ 
denote the support of $P(\cdot|X_{-\infty}^{0}, Y_{-\infty}^{0})$. 
We can compute,
\begin{align*}
\EBig{ 
\frac
{P(Y_{1}^{n}|Y_{-\infty}^{0})}
{P(Y_{1}^{n}|Y_{-\infty}^{0}, X_{-\infty}^{0})}
} 
&= 
\BBE\left(
\EBig{ 
\frac
{P(Y_{1}^{n}|Y_{-\infty}^{0})}
{P(Y_{1}^{n}|Y_{-\infty}^{0}, X_{-\infty}^{0})} 
\Big|Y_{-\infty}^{0}, X_{- \infty }^{0}
}
\right)\\
&=  
\BBE\left(
\sum_{y_{1}^{n} \in B(X_{-\infty}^{0},Y_{-\infty}^{0})}
{ 
\frac
{P(y_{1}^{n}|Y_{-\infty}^{0})}
{P(y_{1}^{n}|Y_{-\infty}^{0}, X_{-\infty}^{0})}
}
P(y_{1}^{n}|Y_{-\infty}^{0}, X_{-\infty}^{0}) 
\right) \\ 
&\leq 1.
\end{align*}
By Markov's inequality,  
\begin{equation*}
\BBP\left[
\log{\Big(\frac{P(Y_{1}^{n}|Y_{-\infty}^{0})}{P(Y_{1}^{n}
|Y_{-\infty}^{0}, X_{-\infty}^{0})} \Big)} > t_n\right]
= \PBig{ \frac{P(Y_{1}^{n}|Y_{-\infty}^{0})}{P(Y_{1}^{n}
|Y_{-\infty}^{0}, X_{-\infty}^{0})} > 2^{t_n}}
\leq 2^{-t_n},
\end{equation*}
and so, by the Borel-Cantelli lemma,
\begin{equation*}
\log{\frac{P(Y_{1}^{n}|Y_{-\infty}^{0})}{P(Y_{1}^{n}|
Y_{-\infty}^{0}, X_{-\infty}^{0})} } \leq t_n, 
\qquad \text{eventually a.s.,}
\end{equation*}
as claimed.
\end{proof}

\medskip

\noindent
{\bf Proof of Theorem~\ref{strongapproxtheorem}. }
Let $K>0$ arbitrary. By Markov's inequality and Kac's theorem,
 \begin{align*}
 \PP(\clR_{n}(\Xp|\Yp) > K \bigm|X_{1}^{n} = x_{1}^{n}, Y_{1}^{n} = y_{1}^{n})
 &\leq \frac{\EBig{\clR_{n}(\Xp|\Yp) \bigm| X_{1}^{n} = x_{1}^{n}, Y_{1}^{n} 
 = y_{1}^{n} }}{K} \\ &= \frac{1}{KP(x_{1}^{n}|y_{1}^{n})}.
 \end{align*}
Taking $K =2^{c_n}/P(X_{1}^{n}|Y_{1}^{n})$, we obtain,
 \begin{align*} 
 &\PP\bigl(\log[\clR_{n}(\Xp|\Yp)P(X_{1}^{n} \bigm|Y_{1}^{n})] > c_n \bigm|X_{1}^{n} = x_{1}^{n}, Y_{1}^{n} = y_{1}^{n}\bigr) \\ 
&= \BBP\Big(\clR_{n}(\Xp|\Yp) > \frac{2^{c_n}}{P(X_{1}^{n} |Y_{1}^{n})}\Big|X_{1}^{n} = x_{1}^{n}, Y_{1}^{n} = y_{1}^{n}\Big) \\
 &\leq 2^{-c_n}.
 \end{align*}
 Averaging over all 
	$x_{1}^{n} \in \mathcal{X}^{n},
	y_{1}^{n} \in \mathcal{Y}^{n}$,
 \begin{equation*}
 \PP\bigl(\log{\clR_{n}(\Xp|\Yp)P(X_{1}^{n} |Y_{1}^{n})} > c_n) \leq 2^{-c_n},
 \end{equation*}
 and the Borel-Cantelli lemma gives~$(i)$.
 
For~$(ii)$ we first note that the probability,
\begin{equation}
\PP\bigl(\log[\clR_{n}(\Xp|\Yp)P(X_{1}^{n}|Y_{1}^{n},
X_{-\infty}^{0},Y_{-\infty}^{0})]
< -c_n \bigm| Y_{1}^{n} = y_{1}^{n}, X_{-\infty}^{0} = x_{-\infty}^{0},Y_{-\infty}^{0} = y_{-\infty}^{0}\bigr) 
\label{eq:pii}
\end{equation}
is the probability,
under
$P(X_1^n=\cdot|
Y_1^n=y_1^n,X_{-\infty}^0=x_{-\infty}^0,Y_{-\infty}^0=y_{-\infty}^0 )$,
of those 
$z_{1}^{n}$
such that,
\ben
P(X_{1}^{n} = z_{1}^{n}|X_{-\infty}^{0},Y_{-\infty}^{n})
< \frac{2^{-c_n}}{\clR_{n}(x_{-\infty}^{0}\ast z_{1}^{n}|y_{-\infty}^{n})},
\een
where `$*$' denotes the concatenation of strings.
Let $G_{n} = G_{n}(x_{-\infty}^{0},y_{-\infty}^{n})\subset\clX^n$
denote the set of all such $z_1^n$.
Then the probability in~(\ref{eq:pii}) is,
\begin{equation*}
\sum_{z_{n} \in G_{n}}{P(z_{1}^{n}|x_{-\infty}^{0},y_{-\infty}^{n})} 
\leq \sum_{z_{n} \in G_{n}} {\frac{2^{-c_n}}
{\clR_{n}(x_{-\infty}^{0}\ast z_{1}^{n}|y_{-\infty}^{n})}} 
\leq 2^{-c_n}\sum_{z_{n} \in \mathcal{X}^{n}} 
{\frac{1}{\clR_{n}(x_{-\infty}^{0}\ast z_{1}^{n}|y_{-\infty}^{n})}}.
\end{equation*}
Since both $x_{-\infty}^{0}$ and $y_{-\infty}^{n}$ are fixed, 
for each $j \geq 1$, there is exactly one $z_{1}^{n} \in \mathcal{X}^{n}$, 
such that $\clR_{n}(x_{-\infty}^{0}\ast z_{1}^{n}|y_{-\infty}^{n}) = j$. 
Thus, the last sum is bound above by,
$$\sum_{j = 1}^{|\mathcal{X}|^{n}}{\frac{1}{j}} \leq Dn,$$ 
for some positive constant $D$.
Therefore, 
the probability in~(\ref{eq:pii}) is bounded above by 
$Dn2^{-c_n}$,
which is independent of $x_{-\infty}^{0},y_{-\infty}^{n}$ and, 
by assumption, summable over $n$. Hence, after averaging over all 
infinite sequences $x_{-\infty}^{0}, y_{-\infty}^{n}$, 
the Borel-Cantelli lemma gives~$(ii)$.

For part~$(iii)$ we have,
eventually, almost surely,
\begin{align*}
\log
&
	\left[
	\clR_{n}(\Xp|\Yp)\frac{P(X_1^n,Y_1^n|Y_{-\infty}^0,X_{-\infty}^0)}
	{P(Y_1^n|Y_{-\infty}^0)}\right]\\
&
	= \log
	\left[
	\clR_{n}(\Xp|\Yp) 
	\frac{P(X_1^n|Y_1^n,X_{-\infty}^{0},Y_{-\infty}^0)P(Y_1^n
	|X_{-\infty}^0,Y_{-\infty}^0)}{P(Y_{1}^{n}|Y_{-\infty}^{0})} 
	\right] \\
&=  
	\log
	[\clR_{n}(\Xp|\Yp)P(X_{1}^{n}|Y_{1}^{n},X_{-\infty}^{0},Y_{-\infty}^{0})]
	+ \log
	\left[\frac{P(Y_{1}^{n}|X_{-\infty}^{0},Y_{-\infty}^{0})}
	{P(Y_{1}^{n}|Y_{-\infty}^{0})}
	\right] \;\geq\;
	-2c_n, 
\end{align*}
where the last inequality follows 
from~$(ii)$ and Lemma~\ref{lemmacover},
and we have shown~$(iii)$.
\qed

\medskip

\noindent
{\bf Proof of Corollary~\ref{strongapproxcorollary}. } 
If we take $c_n = \epsilon n^{\beta}$ in theorem \ref{strongapproxtheorem},
with $\epsilon>0$ arbitrary,
we get from~$(i)$ and~$(iii)$,
\begin{align}
\label{upper}
&\limsup_{n\to\infty}{\frac{1}{n^{\beta}}
\log[\clR_{n}(\Xp|\Yp)P(X_{1}^{n}|Y_{1}^{n})}] \leq 0,
\qquad \mbox{a.s.}\\
\label{lower}
\mbox{and}
\qquad
&\liminf_{n\to\infty}
\frac{1}{n^{\beta}}
\log
\left[\clR_{n}(\Xp|\Yp)\frac{P(X_{1}^{n},Y_{1}^{n}|
X_{-\infty}^{0}, Y_{-\infty}^{0})}{P(Y_{1}^{n}|Y_{-\infty}^{0})}\right]
\geq 0, \qquad 
\mbox{a.s.}
\end{align}
Hence, to prove~$(a)$ it is sufficient to show that,
as $n\to\infty$,
\begin{equation*}
\log{P(X_{1}^{n}|Y_{1}^{n})} - 
\log
\left[\frac{P(X_{1}^{n},Y_{1}^{n}|X_{-\infty}^{0}, Y_{-\infty}^{0})}
{P(Y_{1}^{n}|Y_{-\infty}^{0})}\right]
= O(1), \qquad \mbox{a.s.},
\end{equation*}
which is exactly Lemma~\ref{fsufficient} in Appendix~\ref{philippstoutproof}.

To prove~$(b)$, taking $\beta = 1$ in~\eqref{upper} 
and~\eqref{lower}, it suffices to show that,
\begin{equation*}
\lim_{n\to\infty}
\left\{\frac{1}{n}\log{P(X_{1}^{n}|Y_{1}^{n})} 
- \frac{1}{n}\log\Big(\frac{P(X_1^n,Y_1^{n}|X_{-\infty}^0,Y_{-\infty}^0)}
{P(Y_{1}^{n}|Y_{-\infty}^{0})}\Big)\right\} = 0, \qquad \mbox{a.s.}
\end{equation*}
But the first term converges almost surely to $-H(\Xp|\Yp)$ 
by the Shannon-McMillan-Breiman theorem, as 
in~(\ref{conditionalmcmillan}),
and the 
second term is,
$$-\frac{1}{n}
\sum_{i = 1}^{n}{\log{P(X_{i},Y_{i}|X_{-\infty}^{i-1},Y_{-\infty}^{i-1})}} 
+ \frac{1}{n}\sum_{i = 1}^{n}{\log{P(Y_{i}|Y_{-\infty}^{i-1})}},$$
which, by the ergodic theorem,
converges almost surely to,
\begin{equation*}
-\BBE[\log{P(X_{0},Y_{0}|X_{-\infty}^{0},Y_{-\infty}^{0})}]
+ \BBE[\log{P(Y_{0}|Y_{-\infty}^{0})}]
= H(\Xp,\Yp) - H(\Yp) 
= 
H(\Xp|\Yp).
\end{equation*}
This completes the proof. 
\qed

\end{appendices}

\bibliographystyle{plain}
\def\cprime{$'$}


\end{document}